\documentclass[10pt,conference]{IEEEtran}\IEEEoverridecommandlockouts
\usepackage{epsfig,graphics,subfigure,psfrag,amsmath,cases}
\usepackage{latexsym,amssymb,amsmath,epsfig,subfigure,algorithm,amsthm}
\usepackage{algorithmic}
\usepackage{color}
\usepackage{ctable}
\usepackage{url}
\usepackage{scrtime}
\usepackage{bm}
\usepackage{graphicx}
\usepackage{epstopdf}
\usepackage{units}
\usepackage[font=footnotesize]{caption}
\usepackage{cite} 

  \title{On the Capacity of  SWIPT Systems with  \\ a Nonlinear Energy Harvesting Circuit}
\author{\IEEEauthorblockN{Rania Morsi,  Vahid Jamali, and Robert Schober}}

\author{ \IEEEauthorblockN{Rania Morsi\IEEEauthorrefmark{1}, Vahid Jamali\IEEEauthorrefmark{1}, Derrick Wing Kwan Ng\IEEEauthorrefmark{2}, and Robert Schober\IEEEauthorrefmark{1}}
\IEEEauthorblockA{\IEEEauthorrefmark{1}Friedrich-Alexander-University Erlangen-N\"urnberg (FAU), Germany\\} 
 \IEEEauthorblockA{\IEEEauthorrefmark{2}The University of New South Wales, Australia
    \\}
 \\[-4ex]}

\date{\thistime,\,\today}

\theoremstyle{plain}
\newtheorem{corollary}{Corollary}
\newtheorem{theorem}{Theorem}

\theoremstyle{definition}

\theoremstyle{remark}
\newtheorem{remark}{Remark}

\newcommand{\E}{\mathbb{E}}
\newcommand{\e}{{\rm e}}
\newcommand{\dd}{{\rm d}}

\newcommand{\quotes}[1]{``#1''}
\newcommand{\definedas}{\overset{\underset{\Delta}{}}{=}}

\begin{document}

\maketitle
\begin{abstract}
In this paper, we study information-theoretic limits for simultaneous wireless information and power transfer (SWIPT) systems employing a practical nonlinear radio frequency (RF) energy harvesting (EH) receiver. In particular, we consider a three-node system with one transmitter that broadcasts a common signal to separated information decoding (ID) and EH receivers. Owing to the nonlinearity of the EH receiver circuit, the efficiency of wireless power transfer depends significantly on the waveform of the transmitted signal. In this paper, we aim to answer the following fundamental question: What is the optimal input distribution of the transmit waveform that maximizes the rate of the ID receiver for a given required harvested power at the EH receiver? In particular, we study the capacity of a SWIPT system impaired by  additive white Gaussian noise (AWGN) under average-power (AP) and peak-power (PP) constraints at the transmitter and an EH constraint at the EH receiver. Using Hermite polynomial bases, we prove that the optimal capacity-achieving input distribution that maximizes the rate-energy region is unique and discrete with a finite number of mass points. Furthermore, we show that the optimal input distribution for the same problem without PP constraint is discrete whenever the EH constraint is active and continuous zero-mean Gaussian, otherwise. Our numerical results show that the rate-energy region is enlarged for a larger PP constraint and that the rate loss of the considered SWIPT system compared to the AWGN channel without EH receiver is reduced by increasing the AP budget. 
\end{abstract}

\renewcommand{\baselinestretch}{1}
\large\normalsize
\section{Introduction}
\label{s:introduction}
In addition to their capability of transferring information, radio frequency (RF) signals are  a viable energy source that can charge low-power devices, such as wireless sensors and Internet-of-Things devices. This dual capability of RF signals has recently attracted significant attention to  the study of simultaneous wireless information and power transfer (SWIPT) systems. In \cite{Varshney2008}, Varshney showed that there exists a fundamental tradeoff between the rate of information transfer and power transfer. This tradeoff is characterized by the boundary of the so-called rate-energy region \cite{Rate_energy_MIMO_RuiZhang2013}. In order to fully characterize a SWIPT system, it is  essential to accurately model the wireless power transfer (WPT) component of the system. The RF-based  energy harvesting (EH) receiver of a WPT system  comprises a rectenna, i.e., an antenna followed by a rectifier, which converts the received RF signal into a direct-current (DC) signal that can charge low-power devices. Most of the literature on WPT  assumed a linear RF EH receiver model. However, in practice, RF EH rectifiers are usually composed  of a diode and a capacitor, and as a result, their input-output characteristic is highly nonlinear \cite{Optimum_behaviour_Georgiadis2013,Waveform_design_WPT_Clerckx_2016}. In particular, for high incident RF powers at the rectifier's input, the output DC power saturates due to the diode's reverse breakdown leading to a reduced RF-to-DC conversion efficiency. This saturation behaviour has been modelled in \cite{Letter_non_linear} by a parametric nonlinear EH model that accurately matches  measurements from practical RF EH circuits. 

Owing to the rectifier's nonlinearity, the RF-to-DC conversion efficiency of an RF EH receiver depends not only on the strength of the input RF power at the rectifier, but also on the waveform of the transmitted RF signal \cite{Optimum_behaviour_Georgiadis2013,Waveform_design_WPT_Clerckx_2016}. For example, experiments have shown that signals with high peak-to-average power ratio (PAPR), such as multisine and chaotic signals tend to yield higher DC powers for a given average incident RF power compared to constant envelope signals  \cite{Optimum_behaviour_Georgiadis2013}. This is because compared to a constant-envelope signal, for  e.g. a pulsed high-PAPR signal having the same average power, the rectifier's capacitor charges to a higher peak amplitude leading to a higher output DC voltage during the capacitor's discharge time, see  \cite[Figure 9]{Optimum_behaviour_Georgiadis2013}. This interesting observation has motivated the design of  RF transmit waveforms that maximize the harvested DC power of practical nonlinear EH receivers. For example, in \cite{Waveform_design_WPT_Clerckx_2016}, an analytical nonlinear model of the rectenna is introduced and the amplitudes and phases of a deterministic multisine signal are jointly optimized to maximize the harvested DC power for WPT. It is shown that, while a linear EH model favours single carrier transmission, the nonlinear model favours multi-carrier transmission. Moreover, assuming perfect channel knowledge at the transmitter, the harvested DC power increases linearly with the number of frequency tones \cite{Waveform_design_WPT_Clerckx_2016}.

While the waveform design goal for a pure WPT system is to maximize the harvested DC energy only, for a SWIPT system, the waveform design goal is to simultaneously maximize both the information rate and the harvested energy, i.e., to maximize the rate-energy region.  For a linear EH receiver model, the EH constraint concerns only the second-order moment of the input distribution of the transmit waveform. Hence, for a linear EH model, waveforms with Gaussian input distribution are optimal for  maximizing the information rate of a SWIPT system under an EH constraint. However, for a SWIPT system with a nonlinear EH circuit, a fundamental open question that naturally arises for the waveform design problem   is \quotes{\emph{What is the optimal input distribution of the transmit waveform that maximizes the rate-energy region?}} First steps toward answering this question have been made in \cite{SWIPT_Clerckx_OFDM_Multisine2016} and \cite{SWIPT_Clerckx_Single_Carrier2017}. In \cite{SWIPT_Clerckx_OFDM_Multisine2016}, the authors considered
the superposition of a deterministic multisine waveform and a modulated orthogonal frequency division multiplexing waveform and optimized the amplitudes  and phases of the frequency tones to maximize the rate-energy region. Furthermore,  in \cite{SWIPT_Clerckx_Single_Carrier2017},  input distributions  that are fully characterized by their first- and second-order statistics were considered and a truncated Taylor series expansion of the diode's nonlinear characteristic equation was adopted. It was shown that the optimal input distribution under these assumptions  is the zero-mean complex Gaussian distribution with asymmetric power allocation to the real and  imaginary parts. However, the waveforms reported in \cite{SWIPT_Clerckx_OFDM_Multisine2016} and \cite{SWIPT_Clerckx_Single_Carrier2017} are not optimal if the imposed restrictions on the input distributions are removed

In this paper, we aim to answer the above fundamental question without imposing restrictions on the input distribution and adopting the exact form of the diode's  characteristic equation. In particular, we consider a three-node SWIPT system, in which one receiver harvests energy and another separate receiver decodes information from a signal broadcasted by a common transmitter over an additive white Gaussian noise (AWGN) channel. For the EH receiver, we adopt the nonlinear rectenna circuit model from \cite{Waveform_design_WPT_Clerckx_2016,Waveform_optimization_SPAWC_Rui_Zhang_2017}.  Our objective is to find the optimal input distribution of the transmit signal that maximizes the mutual information between the input and the output of the information channel under a minimum harvested power constraint at the EH receiver. Thereby, we impose average-power (AP) and peak-power (PP) constraints at the transmitter. 
We show that the capacity-achieving input distribution that maximizes the rate-energy region under these constraints is unique. Furthermore, we show that this optimal input distribution is \emph{discrete} even if the PP constraint is removed. With a PP constraint, the optimal input distribution possesses a finite number of mass points. This interesting outcome is inline with the classical result for the capacity-achieving input distribution of a discrete-time memoryless channel under AP and PP constraints studied by Smith in \cite{SMITH19712}.

\section{System Model and Preliminaries}
\label{s:system_model}
\subsection{System Model}
\begin{figure}[!t] 
\centering
\includegraphics[width=0.3\textwidth]{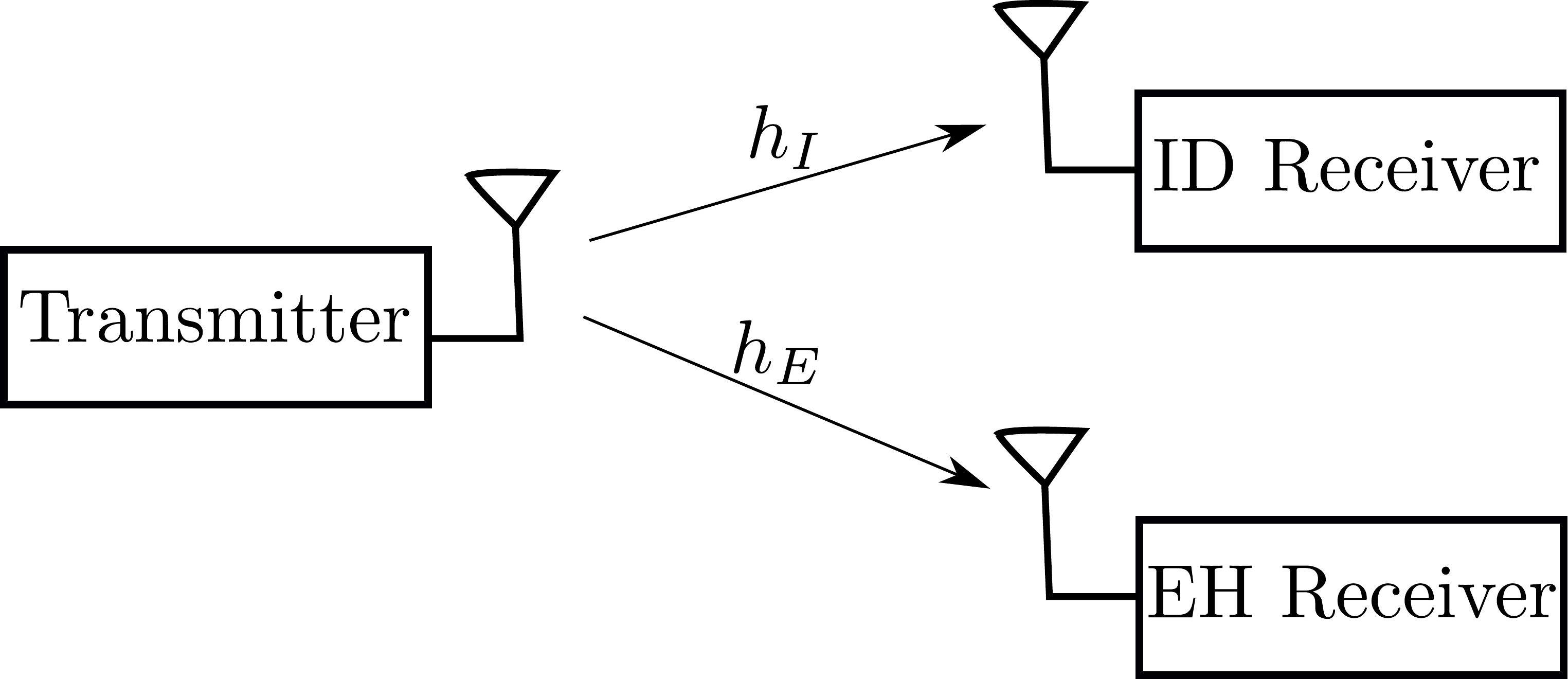}
\caption{A SWIPT system with a separate EH receiver and an ID receiver.}
\label{fig:system_model}
\end{figure}
We consider a three-node single-antenna SWIPT system as shown in Fig. \ref{fig:system_model}, where a transmitter broadcasts a common signal to an information decoding (ID) receiver and an EH receiver. In particular, we consider a time-slotted system with time slot duration $T$. The transmitter transmits a real-valued baseband information-bearing pulse-amplitude modulated signal $x(t)=\sum_{k=-\infty}^{\infty}x[k]g(t-kT)$, where $g(t)$ is the transmit pulse waveform and $x[k]$ is the information symbol in time slot $k$ which is a realization of an independent and identically distributed real-valued\footnote{As customary for capacity analysis, see e.g. \cite{Cover,Nikola_FD}, as a first step,  we  assume real-valued channel inputs and outputs.
The generalization to a complex-valued signal model is relatively straightforward \cite{Cover} but omitted here due to space constraints.} random variable $X\in \mathbb{R}$ having cumulative distribution function $F$. The channel fading gains for  the ID and EH receivers are denoted  by $h_I\in\mathbb{R}$ and $h_E\in\mathbb{R}$, respectively, and are assumed to be flat and fixed over all time slots. Both channel gains are assumed to be perfectly known at the transmitter and the information channel is known at the ID receiver.
The received signal at the ID receiver is $y_I(t)\!=\!x(t)h_I+n(t)$ where $n(t)$ is real-valued zero-mean AWGN.
 At the EH receiver, we ignore the additive noise since its contribution to the harvested DC power is negligible.  Hence, the received signal is $y_E(t)\!=\!x(t)h_E$ in the baseband and $y_E^{\rm RF}(t)\!=\!\sqrt{2}\Re\{y_E(t)\e^{j 2\pi f_c t}\}$ in the RF domain, where $j$ is the imaginary unit, $f_c$ is the carrier frequency, and $\Re\{\cdot\}$ denotes the real part of a complex number.  Next, we focus on the EH receiver and adopt the nonlinear rectenna model from \cite{Waveform_optimization_SPAWC_Rui_Zhang_2017}. In particular, we obtain an expression for the harvested energy at the EH receiver in terms of the input distribution for the considered scalar single-carrier AWGN channel.

\subsection{Rectenna Nonlinear Circuit Model}
\label{ss:rectenna_model} \vspace{-0.1cm}
We adopt the nonlinear rectenna model from \cite{Waveform_design_WPT_Clerckx_2016} and \cite{Waveform_optimization_SPAWC_Rui_Zhang_2017} shown in Fig.~\ref{fig:rectenna_model}.  A rectenna consists of an antenna and a rectifier. The antenna is commonly modelled as an equivalent voltage source $v_s(t)$ in series with an impedance $R_{ant}$. The rectifier typically consists of a single diode followed by a capacitor-based low pass filter (LPF). The received RF signal $y_E^{\rm RF}(t)$ is converted at the rectifier's output to a DC signal across a load resistance $R_L$. We assume perfect impedance matching, i.e., $R_{in}=R_{ant}$ holds, where $R_{in}$ is the equivalent input impedance of the circuit observed after the antenna impedance. Thereby, the average received power is completely transferred to the rectifier, i.e., $\E[|y_E^{\rm RF}(t)|^2]=\E[|v_{in}(t)|^2]/R_{in}$, or equivalently $v_{in}(t)=y_E^{\rm RF}(t)\sqrt{R_{ant}}$, where $\E[\cdot]$ denotes the expectation operator  \cite{Waveform_design_WPT_Clerckx_2016}. The current $i_d(t)$ flowing through an ideal diode is related to the voltage drop, $v_d(t)$, across it by the Shockley diode equation $i_d(t)=i_s\big(\e^{\frac{v_d(t)}{\eta V_T}}-1\big)$, where $i_s$ is the diode's reverse bias saturation current, $\eta$ is the ideality factor which typically lies between 1 and 2, and $V_T$ is the thermal voltage which is approximately $25.85\,$ mV at room temperature. 
By assuming that the capacitance, $c$, of the LPF is sufficiently large, the output voltage can be assumed constant (DC), i.e., $v_{out}(t)\approx v_{out}$ \cite{Waveform_optimization_SPAWC_Rui_Zhang_2017}. Applying Kirchoff's current law to the circuit in Fig. \ref{fig:rectenna_model}, we obtain
\begin{IEEEeqnarray}{ll}
i_d(t)&=i_c(t)+i_{out}(t)=c\frac{\dd v_{out}}{\dd t}+\frac{v_{out}}{R_L} = i_s\big(\e^{\frac{v_d(t)}{\eta V_T}}-1\big),\quad
\end{IEEEeqnarray}
where $v_{out}$ is constant, i.e., $\frac{\dd v_{out}}{\dd t}=0$. Now, using $v_d(t)=v_{in}(t)-v_{out}=y_E^{\rm RF}(t)\sqrt{R_{ant}}-v_{out}$, 
we obtain 
\begin{IEEEeqnarray}{ll}
\frac{v_{out}}{R_L}=i_s\Big(\e^{\frac{-v_{out}}{\eta V_T}}\e^{\frac{y_E^{\rm RF}(t)\sqrt{R_{ant}}}{\eta V_T}}-1\Big),\quad
\end{IEEEeqnarray}
which can be written as $\e^{B y_E^{\rm RF}(t)}=\big(1+\frac{v_{out}}{i_sR_L}\big)\e^{\frac{v_{out}}{\eta V_T}}$, where $B=\frac{\sqrt{R_{ant}}}{\eta V_T}$. Finally, averaging both sides over one symbol duration $T$ and over the input distribution $F$, we obtain  \cite{Waveform_optimization_SPAWC_Rui_Zhang_2017}
\begin{equation}
\E\left[\frac{1}{T}\int_{T}\e^{B y_E^{\rm RF}(t)}\dd t\right]=\left(1+\frac{v_{out}}{i_sR_L}\right)\e^{\frac{v_{out}}{\eta V_T}}.
\label{eq:prop_to_EH}
\end{equation}
 The left hand side (LHS) of (\ref{eq:prop_to_EH}) can be interpreted as the time-average of the moment generating function $\E[\e^{B y_E^{\rm RF}(t)}]$ of random variable $y_E^{\rm RF}(t)$, which is the received RF signal at time instant $t$. 
We note that the DC power delivered to the load is $p_{out}=v_{out}^2/R_L$ and the right hand side of (\ref{eq:prop_to_EH}) strictly increases with $v_{out}$. Hence, imposing a minimum harvested power constraint $p_{out} \geq p_{req}$ is equivalent to imposing a constraint, $E_{req}$, on (\ref{eq:prop_to_EH}), i.e., 
\begin{IEEEeqnarray}{ll}\label{eq:E_req_def}
\E\left[\frac{1}{T}\int_{T}\e^{B y_E^{\rm RF}(t)}\dd t\right] \geq E_{req}\triangleq 
\left(1+\frac{\sqrt{p_{req}}}{i_s\sqrt{R_L}}\right)\e^{\frac{\sqrt{R_Lp_{req}}}{\eta V_T}}.\quad
\end{IEEEeqnarray}
\begin{figure}[!t] 
\centering
\includegraphics[width=0.4\textwidth]{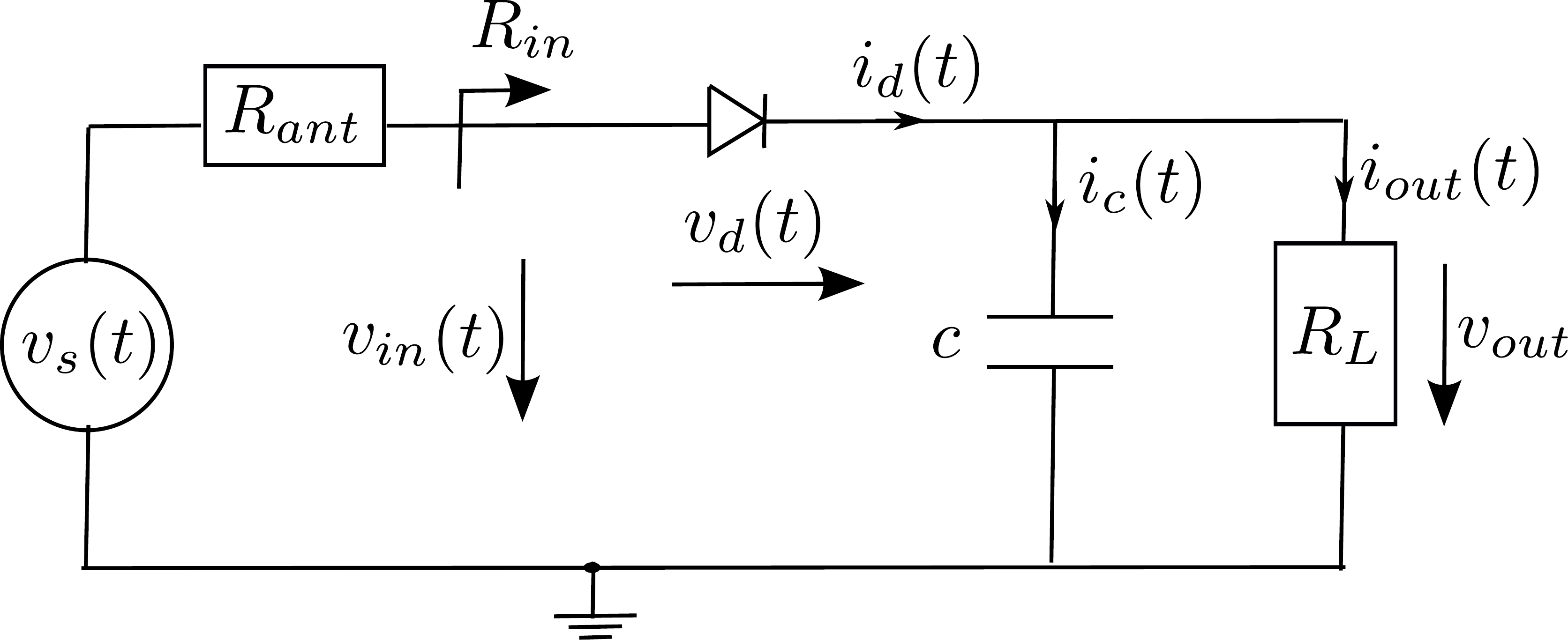}
\caption{Nonlinear rectenna circuit model.\\[-1ex]}
\label{fig:rectenna_model}
\end{figure}
Assuming a rectangular pulse $g(t)$ with unit amplitude and duration $T$, in time slot $k$, i.e., $kT\!-\!T/2\!<\!t\!<\!kT\!+\!T/2$, the baseband transmit signal is constant and given by $x(t)\!=\!\sum_{k=-\infty}^{\infty}x[k]g(t-kT)\!=\!x[k]$. Hence, the received signal in the RF domain  reduces to $y_E^{\rm RF}(t)\!=\!\sqrt{2} x[k] h_E\cos(2\pi f_c t)$, $kT\!-\!T/2\!<\!t\!<\!kT\!+\!T/2$, where the information symbol $x[k]$ is a realization of  random variable $X$ at time slot $k$. Hence, (\ref{eq:prop_to_EH}) can be written as
\begin{IEEEeqnarray}{ll}
\E\left[\frac{1}{T}\!\int_{T}\e^{B y_E^{\rm RF}(t)}\dd t\right]\!&=\!\E\left[\frac{1}{T}\!\int_{T}\e^{\sqrt{2}B X h_E \cos(2\pi f_c t)} \dd t \right] \nonumber\\
&=\E\left[I_0\left(\sqrt{2}B h_E X \right)\right],
\label{eq:EH_eqn}
\end{IEEEeqnarray}
where $I_0(\cdot)$ is the modified Bessel function of the first kind and order zero. In solving the integral in (\ref{eq:EH_eqn}), we assumed that $f_c=m/T$ with integer $m$. 
Using (\ref{eq:EH_eqn}), the EH constraint reduces to
\begin{equation}
\E\left[I_0\left(\sqrt{2}B h_E X \right)\right] \geq E_{req}.
\label{eq:EH_Constraint}
\end{equation}
In particular, for a given required harvested power $p_{req}$, $E_{req}$ is calculated in (\ref{eq:E_req_def}) and the EH constraint in (\ref{eq:EH_Constraint}) is applied in the optimization problem for the optimal input distribution formulated in the next section. Then, with the optimized input distribution of $X$, the LHS of (\ref{eq:EH_Constraint}) is evaluated, which based on (\ref{eq:EH_eqn}) is equivalent to the LHS of (\ref{eq:prop_to_EH}). Using (\ref{eq:prop_to_EH}), $v_{out}$ can be obtained using the bisection method and the  DC power delivered to the load is obtained as $p_{out}=v_{out}^2/R_L$.

In addition to the EH constraint in (\ref{eq:EH_Constraint}), the peak amplitude of the received RF signal $y_E^{\rm RF}(t)$ should be  limited to avoid the breakdown of the rectifying diode \cite{SWIPT_Clerckx_OFDM_Multisine2016}. In particular, we limit the peak amplitude at the EH receiver input to $A_R$, i.e., $\max| y_E^{\rm RF}(t)|=\max|\sqrt{2}Xh_E|\leq A_{R}$, which consequently limits the amplitude of the transmit signal $X$ to $A_R/(\sqrt{2}h_E)$.
\section{Problem Formulation and Solution}
\label{s:Problem_formulation}
In this section, we study the capacity of the considered scalar AWGN information channel under AP and PP  constraints on the transmit signal and an EH constraint at the EH receiver. We first establish the existence of a unique optimal input distribution for the transmitted  information symbols. We then provide necessary and sufficient conditions for the capacity-achieving input distribution. Furthermore, we show that the capacity-achieving input distribution is discrete with a finite number of mass points. \vspace{-0.1cm}
\subsection{Problem Formulation}\vspace{-0.1cm}
The discrete-time baseband model for the information channel after down-conversion, matched filtering, and sampling of the continuous-time signal received at the ID receiver is given by $Y=X h_I+N$, where $N\sim\mathcal{N}(0,\sigma_n^2)$ is the Gaussian distributed noise and $Y$ is the information channel output with probability density function (pdf) $p(y)$. At the transmitter, the peak power is usually limited to avoid the negative impact of amplifier nonlinearities. Hence, we set a maximum transmit amplitude constraint given by $|X|\leq A_{T}$. Recall that $|X|$ is also limited by $A_R/(\sqrt{2}h_E)$ in order to avoid the breakdown of the rectifying diode, cf. Section \ref{ss:rectenna_model}. Hence, the effective amplitude constraint on the transmit signal, i.e., the effective PP constraint, reduces to $|X|\leq\min(A_T,A_R/(\sqrt{2}h_E))\definedas A$. We aim at maximizing the average mutual information between $X$ and $Y$ subject to AP and PP constraints on the transmit symbols $X$ and a minimum harvested power constraint at the EH receiver.  Hence, our optimization problem can be formulated as\vspace{-0.1cm}
\begin{IEEEeqnarray}{llll}
C=&\sup\limits_{F\in\mathcal{F}_A}& &I(F) \nonumber\\
&\,\,{\rm s.t.} &&{\rm C1:}\quad \E[X^2] \leq \sigma_x^2
\nonumber\\
&&&{\rm C2:}\quad \E\left[I_0\left(\sqrt{2}B h_E X \right)\right] \geq E_{req},
\label{eq:capacity_problem}
\end{IEEEeqnarray}
where $\mathcal{F}_A$ is the set of all possible input distribution functions of  random variable $X$ that satisfy the PP constraint $|X|\!\leq\!A$, i.e., $\forall F\!\in\!\mathcal{F}_A$, $\int_{-A}^A\dd F(x)\!=\!1$. $I(F)$ is the mutual information between $X$ and $Y$ achieved by the input distribution $F$ and given by $I(F)=\int_{-A}^A i(x;F)\dd F(x)$, where $i(x;F)$ is the marginal information density defined as $i(x;F)\definedas\int_yp(y|x) \log_2\frac{p(y|x)}{p(y;F)}\dd y$, $p(y;F)$ is the output pdf assuming input distribution $F$, and $p(y|x)$ is the  output pdf conditioned on the transmission of symbol $x$ \cite{SMITH19712}.  $\sigma_x^2$ is the AP budget and $E_{req}$ is the EH constraint, cf. Section \ref{ss:rectenna_model}. For the purpose of exposition, we define $g_1(F)\definedas\int_{-A}^A x^2\dd F(x) - \sigma_x^2$ and $g_2(F)\definedas E_{req}-\int_{-A}^A I_0\left(\sqrt{2}B h_E x\right)\dd F(x)$. Hence, constraints C1 and C2 can be written as $g_i(F)\leq 0$, $i=1,2$. Note that solving (\ref{eq:capacity_problem}) for all possible targeted harvested energies at the EH receiver, $p_{req}$, or the corresponding $E_{req}$ in C2, leads to the rate-energy region of the considered SWIPT system. \vspace{-0.1cm}
\subsection{Properties of the Optimal Input Distribution}
In the following, we investigate some important properties of the optimal input distribution.
\subsubsection{Uniqueness of the Optimal Input Distribution}
We establish the uniqueness of the optimal input distribution for problem (\ref{eq:capacity_problem}) in the following theorem. 
\begin{theorem}\normalfont
The capacity $C$ in (\ref{eq:capacity_problem}) is achieved by a \emph{unique} optimal input distribution function $F_0$, i.e., $C=\sup\limits_{F\in\Omega}I(F)=I(F_0)$, where $\Omega\subset \mathcal{F}_A$ is the set of input distributions that satisfy the PP constraint and constraints C1 and C2 in (\ref{eq:capacity_problem}). Furthermore, there exist $\lambda_1\geq 0$ and $\lambda_2\geq 0$ such that the capacity $C$ is equivalently given by\vspace{-0.2cm}
\begin{equation}
C=\sup\limits_{F\in\mathcal{F}_A} \quad I(F)-\lambda_1 g_1(F) -\lambda_2 g_2(F).\vspace{-0.2cm}
\label{eq:C_lagrangian}
\end{equation}
Moreover, the supremum in (\ref{eq:C_lagrangian}) is also achieved by $F_0$ and
$\lambda_1g_1(F_0)=0$ and $\lambda_2g_2(F_0)=0$.
\label{theo:unique_distribution}
\end{theorem}
\begin{proof}
The proof is provided in Appendix \ref{App:proof_unique_optimal_distribution}.
\end{proof}\vspace{-0.1cm}
\subsubsection{Necessary and Sufficient Conditions for the Optimal Input Distribution}
The following theorem provides a necessary and sufficient condition for the capacity-achieving distribution $F_0$.
\begin{theorem}\normalfont
A necessary and sufficient condition for the input distribution $F_0$ to achieve the capacity $C$ in (\ref{eq:C_lagrangian}) is that $\forall F\in\mathcal{F}_A$, there exist $\lambda_1\geq 0$ and $\lambda_2\geq 0$ such that \vspace{-0.3cm}
\begin{IEEEeqnarray}{ll}
\int\limits_{-A}^A \left[i(x;F_0)-\lambda_1 x^2+\lambda_2I_0(\sqrt{2} B h_E x)\right]\dd F(x) \vspace{-0.2cm}\nonumber\\
\qquad\qquad \leq C-\lambda_1\sigma_x^2+\lambda_2 E_{req}.\label{eq:C_necessary_sufficient_condition1}
\vspace{-0.5cm}
\end{IEEEeqnarray}
\label{theo:C_necessary_sufficient_condition1}
\end{theorem}\vspace{-0.4cm}
\begin{proof}
The proof is provided in Appendix \ref{App:proof_C_necessary_sufficient_condition1}.
\end{proof}
Let us define the points of increase of a distribution function $F$ as those points which have non-zero probability \cite{SMITH19712}.
In the following corollary, we use the condition in (\ref{eq:C_necessary_sufficient_condition1}) to provide a more useful set of necessary and sufficient conditions for the optimal input distribution.
\begin{corollary}\normalfont
Let $E_0$ be the points of increase of a distribution function $F_0$ on $[-A,A]$, then $F_0$ is the optimal input distribution if and only if there exist $\lambda_1\geq 0$ and $\lambda_2\geq 0$ such that 
\begin{equation}
\begin{aligned}
&\lambda_1\!\left(x^2-\sigma_x^2\right)\!-\!\lambda_2\left(I_0\left(\sqrt{2}Bh_E x\right)\!-\!E_{req}\right)+C\\
&\!+\!\frac{1}{2}\!\log_2(2\pi\e\sigma_n^2)\!+\!\!\frac{1}{\sqrt{2\pi\sigma_n^2}}\!\!\int\!\!\!\e^{-\frac{(y-xh_I)^2}{2\sigma_n^2}}\log_2(p(y;F_0))\dd y \geq 0,
\end{aligned}
\label{eq:C_necessary_sufficient_condition2}
\end{equation}
$\forall x\in[-A,A]$, with equality if $x$ is a point of increase of $F_0$, i.e., if $x\in E_0$.
\label{corol:C_necessary_sufficient_condition2}
\end{corollary}
\begin{proof}
The proof is provided in Appendix \ref{App:proof_C_necessary_sufficient_condition2}.
\end{proof}

\subsubsection{Discreteness of the Optimal Input Distribution}
The discreteness of the optimal input distribution $F_0$ for the problem in (\ref{eq:capacity_problem})   is formally stated in the following theorem. 
\begin{theorem}\normalfont
The optimal input distribution that achieves the capacity in (\ref{eq:capacity_problem}) is discrete with finite number of mass points.
\label{theo:discrete_optimal_distribution}
\end{theorem}
\begin{proof}
The proof is provided in Appendix \ref{App:proof_discrete_optimal_distribution}.
\end{proof}
\begin{remark}
If the EH constraint in problem (\ref{eq:capacity_problem}) is inactive, then the problem reduces to the capacity of an AP and PP constrained AWGN channel, whose optimal input distribution was shown to be discrete with a finite number of mass points by Smith in \cite{SMITH19712}. 
\label{rem:AP_PP_only}
\end{remark}

\begin{corollary}\normalfont
Consider problem (\ref{eq:capacity_problem}) without the PP constraint, i.e., $A\to\infty$ and $X\in\mathbb{R}$ and define $E_{lim}\definedas \e^{B^2h_E^2\sigma_x^2/2} I_0(B^2h_E^2\sigma_x^2/2)$. If $E_{req}\leq E_{lim}$, then the EH constraint is inactive and problem (\ref{eq:capacity_problem}) reduces to the capacity maximization problem of the AP-constrained AWGN channel whose optimal input distribution is known to be the \emph{continuous} zero-mean Gaussian distribution \cite{Cover}.  On the other hand, if $E_{req}>E_{lim}$, then the EH constraint is active and the optimal input distribution of the AWGN channel with AP and EH constraints is discrete. 
\label{coro:AP_EH_only}
\end{corollary}
\begin{proof}
For problem (\ref{eq:capacity_problem}) without PP  constraint,  if the EH constraint C2 is satisfied for a zero-mean Gaussian distribution, i.e., $E_{lim}=\int_{-\infty}^{\infty} I_0(\sqrt{2}Bh_E x)\frac{1}{\sqrt{2\pi\sigma_x^2}}\e^{-x^2/(2\sigma_x^2)}\dd x\!=\!\e^{B^2h_E^2\sigma_x^2/2} I_0(B^2h_E^2\sigma_x^2/2)\!\geq \!E_{req}$ holds, then C2 is inactive and the continuous Gaussian distribution is optimal \cite{Cover}. This can also be verified from Case 1 in Appendix D with $A\!\to\!\infty$. Otherwise, if $E_{req}\!>\!E_{lim}$, then the problem is infeasible with the Gaussian distribution and the optimal input distribution that satisfies the EH constraint is discrete. This is because the proof of the discreteness of the optimal input distribution in Case 2 in Appendix \ref{App:proof_discrete_optimal_distribution}  is independent of the value of $A$.
\end{proof}

\subsection{Optimal Input Distribution as the Solution of (\ref{eq:capacity_problem})}

Here, we explain how the optimal input distribution is obtained from (\ref{eq:capacity_problem}). Note that although we showed that the optimal input distribution is discrete with a finite number of  mass points, cf. Theorem~\ref{theo:discrete_optimal_distribution}, the number and positions of the mass points are not known. To cope with this issue, we discretize the interval $[-A,A]$ with sufficiently small step size $\Delta x$ to obtain the symbol set. Then, we employ a numerical solver such as CVX \cite{CVX} to solve  (\ref{eq:capacity_problem}). Since the optimization problem in  (\ref{eq:capacity_problem}) is convex, the global optimum solution can be found using standard numerical methods and letting $\Delta x\to 0$.
\section{Numerical Results}
\label{s:numerical_results}
In this section, we numerically evaluate  problem (\ref{eq:capacity_problem}) to obtain the capacity of  the considered AWGN channel with AP, PP, and EH constraints. The path loss model is given by $h_r^2=\left(\frac{c_l}{4\pi d_r f_c}\right)^\alpha$ for  $r\in\{I,E\}$, where $c_l$ is the speed of light, $f_c$ is the carrier frequency, $\alpha$ is the path loss exponent, $d_I$ and $d_E$ are the distances between the transmitter and the ID and EH receivers, respectively. We consider a setup with $f_c=2.45\,$GHz, $\alpha=2.5$, $d_I=500\,$m, and $d_E=70\,$m. At the ID receiver, we assume a noise power of $\sigma_n^2=-80\,$dBm. At the EH receiver, we assume the following circuit parameters $R_{ant}=50\,\Omega$, $i_s=100\,\mu$A, $\eta=1.5$, $V_T=25.85\,$mV, and $R_L=10\,$k$\Omega$ \cite{Waveform_design_WPT_Clerckx_2016,Waveform_optimization_SPAWC_Rui_Zhang_2017}.

In Fig. \ref{fig:RE_region_diff_EH_diff_PP_diff_AP}, we plot the rate-energy region of the considered system for different AP and PP constraints. In particular, we obtain the rate-energy region by solving optimization problem (\ref{eq:capacity_problem}) for a given $E_{req}$ corresponding to the targeted harvested DC power $p_{out}$. 
It is observed that, for the considered separated ID and EH receivers with the nonlinear EH model, there is a tradeoff between the information rate transmitted to the ID receiver and the power delivered to the EH receiver. This is because, for a larger required harvested power, the optimal input distribution is such that the transmitter transmits more often with the peak amplitudes $x=\pm A$ and less often in the range  $x\in(-A,A)$. This leads to higher harvested power for the EH receiver at the expense of a lower information rate for the ID receiver. Moreover, it can be observed that the higher the peak-amplitude $A$, the larger the rate-energy region gets. This is because, for a larger peak amplitude, the transmitter has to transmit less often with the peak amplitudes and can more often choose $x\in(-A,A)$ allowing for a higher information rate. In addition, we plot Shannon's capacity limit given by $C=0.5\log_2(1+\sigma_x^2h_I^2/\sigma_n^2)$, which is the capacity of the AWGN channel with AP constraint only. We also plot the capacity of the AWGN channel with AP and PP constraints studied by Smith in \cite{SMITH19712}. For the considered AP constraints, Shannon's and Smith's capacities practically coincide due to the low APs. For this reason, for a given AP constraint, all rate-energy curves for different PP constraints converge to the same point (Shannon's capacity) when the EH constraint is inactive. In other words, if we define $p_{lim}$ as the maximum harvested DC power obtained by setting (\ref{eq:EH_eqn}) to $E_{lim}$, then from  Corollary \ref{coro:AP_EH_only},  when the required harvested power is strictly less than $p_{lim}$, Shannon's capacity is achieved and the harvested DC power with the optimal Gaussian input distribution is $p_{lim}$. 

\begin{figure}[!t] 
\centering\hspace{-0.5cm}
\includegraphics[width=0.51\textwidth]{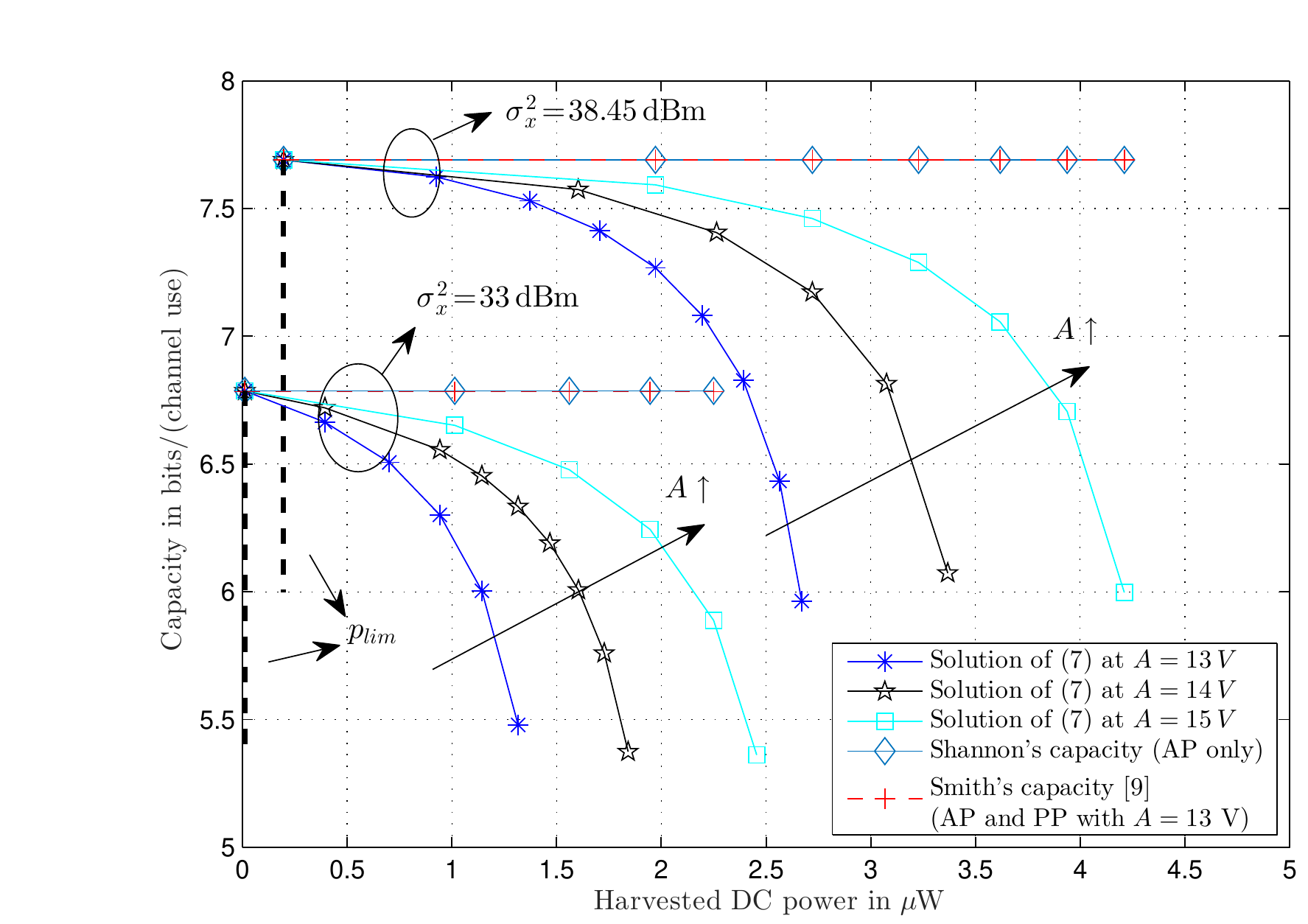}
\caption{Rate-energy region for different AP and PP  constraints.}
\label{fig:RE_region_diff_EH_diff_PP_diff_AP}
\end{figure}

In Fig. \ref{fig:Different_AP_optimal_finite_Shannon_Smith}, we 
plot the capacity according to problem (\ref{eq:capacity_problem})  as a function of  AP constraint $\sigma_x^2$ for a required harvested DC power of $3\,\mu$W and a peak amplitude of $A=13\,$V. For low APs (or equivalently low signal-to-noise ratio (SNR)), the system is EH-limited. In particular, compared to Smith's problem in \cite{SMITH19712} with AP and PP constraints, the imposed EH constraint of our problem in (\ref{eq:capacity_problem}) incurs a capacity loss which decreases with the AP. On the other hand, for large APs  (SNRs), the system is PP limited. That is, the EH constraint is inactive and the capacity of our problem coincides with that of Smith's problem in \cite{SMITH19712}. In addition, we plot the maximum information rate for amplitude shift keying (ASK) modulation. This rate is obtained by solving problem (\ref{eq:capacity_problem}) for symbols $x=\frac{2Ak}{M-1}-A,\,\,k=0,1,\dots,M-1$, where $M$ is the number of symbols. The larger the alphabet size, the closer the capacity achieved by the finite alphabet is to that achieved by the optimal input distribution with arbitrary number of mass points. Moreover, we observe that in the PP-limited regime, the capacities of all PP-constrained schemes saturate with increasing AP budget. 
\begin{figure}[!t] 
\centering\hspace{-0.5cm}
\includegraphics[width=0.51\textwidth]{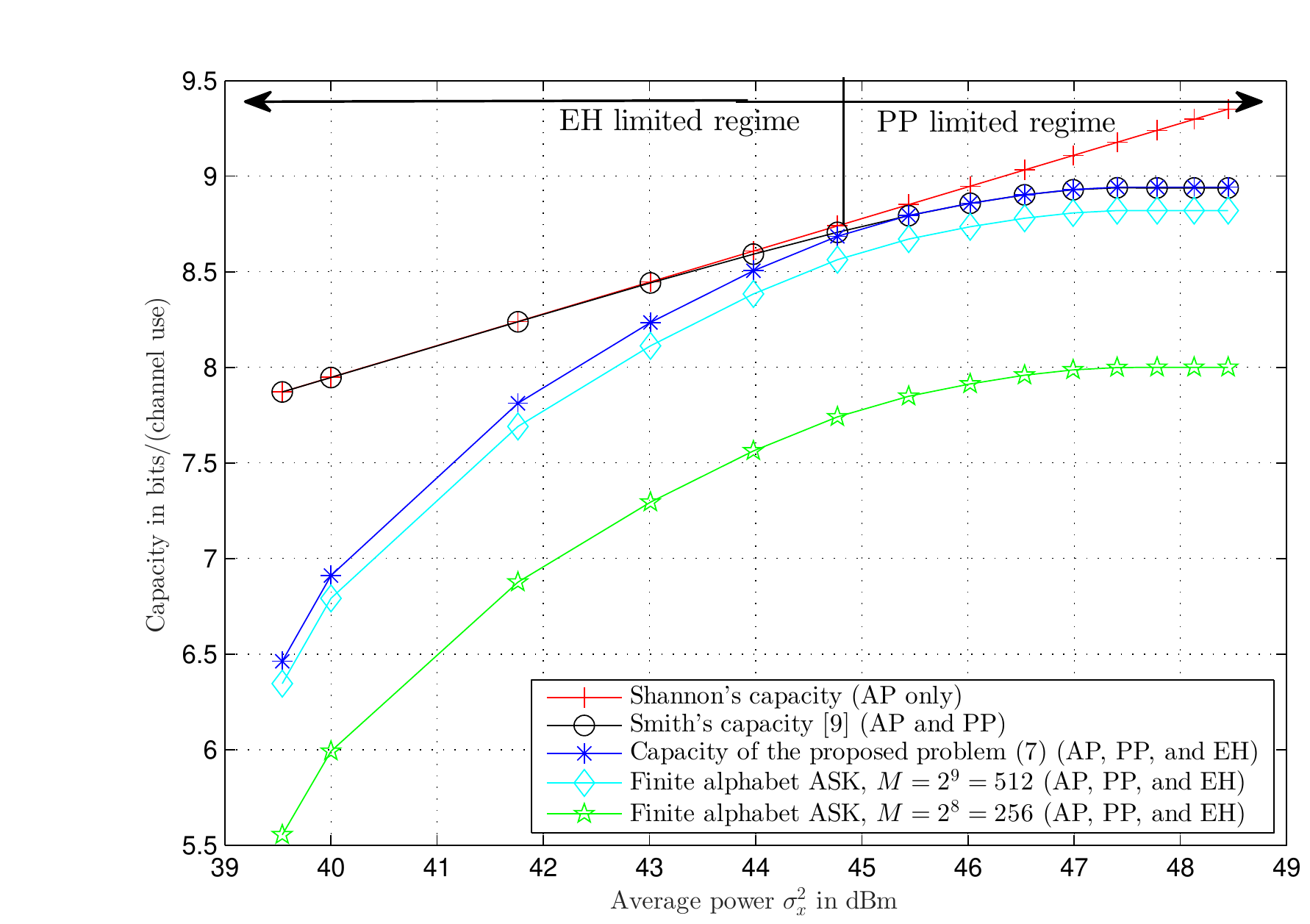}
\caption{Capacity of problem (\ref{eq:capacity_problem})  and for different finite-alphabet sizes for $A=13\,$V and a required harvested DC power of $3\,\mu$W.} 
\label{fig:Different_AP_optimal_finite_Shannon_Smith}
\end{figure}
\section{Conclusion}
\label{s:conclusion}
In this paper, we considered a practical nonlinear RF EH model and studied the AWGN channel capacity of a SWIPT system with separated ID and EH receivers under AP, PP, and EH constraints. We  showed that the capacity-achieving optimal input distribution  that maximizes the rate-energy region  is discrete with a finite number of mass points which is inline with the results obtained for information-only transfer systems in \cite{SMITH19712,capacity_Rayleigh_fading_Abou_Faycal2001,Hermite_bases_Abou_Faycal_2012}.   Furthermore, we proved that the optimal input distribution for the same problem without PP constraint is discrete whenever the EH constraint is active and is continuous zero-mean Gaussian, otherwise. Moreover, we showed that the rate-energy region increases if the PP constraint is relaxed and that the loss in capacity incurred by the EH constraint decreases as the AP budget increases.

\appendices
\section{Proof of Theorem \ref{theo:unique_distribution}}
\label{App:proof_unique_optimal_distribution}
We first prove the existence of a unique distribution $F_0 \in\Omega$ that maximizes the mutual information $I(F)$. It suffices to show that the optimization problem in (\ref{eq:capacity_problem}) is convex, i.e., that the set $\Omega$ is convex and compact in some topology and that $I(\cdot)$ is continuous and strictly concave in $F$.  The convexity of the set $\Omega$ follows from the convexity of the set of distribution functions $\mathcal{F}_A$ (defined by $\int_{-A}^A \dd F(x)=1$) and the linearity of the AP and  EH constraints in $F$.  Hence, constraints $g_i(F)\leq 0$, $i=1,2$, are convex. 
The proof of the compactness of $\Omega$ is similar to that in \cite[Appendix I.A]{capacity_Rayleigh_fading_Abou_Faycal2001}. Next, we show that the mutual information is continuous and strictly concave in $F$. The mutual information resulting from an input distribution $F$ is given by $I(F)=h_Y(F)-h_N$, where $h_Y(F)$ is the entropy of  output $Y$ assuming an input distribution function $F$,  and $h_N$ is the noise entropy which is constant for the considered AWGN channel and given by $h_N=\frac{1}{2}\log_2(2\pi\e\sigma_n^2)$. Since $h_N$ is constant, it suffices to show that $h_Y(F)$ is continuous and strictly concave. The proof of the continuity of $h_Y(F)$ is given in \cite[Appendix I.B]{capacity_Rayleigh_fading_Abou_Faycal2001}. Next, we show that the entropy function $h_Y(F)$ is strictly concave in $F$. Since $h_Y(F)=-\int_{-\infty}^{\infty}p(y; F)\log_2(p(y; F)) \dd y$ is a strictly concave function of the output pdf $p(y;F)$ and $p(y;F)=\int_{-\infty}^{\infty}p(y|x)\dd F(x)$ is a linear function in $F$, it follows that $h_Y(F)$ is a strictly concave function of $F$. Hence, we conclude that problem (\ref{eq:capacity_problem}) is convex and has a unique solution. 

Next, the proof that the capacity $C=\sup_{F\in\Omega}I(F)$ is equivalently given by (\ref{eq:C_lagrangian}) follows from the Lagrangian theorem for constrained optimization. In particular, this equivalence (strong duality) holds for the convex problem in (\ref{eq:capacity_problem}) if $C$ is finite and Slater's condition holds, i.e., there exists an interior point $F\in\mathcal{F}_A$ such that all constraints hold with strict inequality, i.e., $g_i(F)<0$, $i=1,2$. The finiteness of the capacity $C$ is guaranteed by the AP  constraint. Next, we prove that for the considered problem, Slater's condition holds. Let $x_1$ satisfies $|x_1|<\sigma_x<A$ and $I_0\left(\sqrt{2}B h_E x_1\right)>E_{req}$ and let $F_1$ be the unit-step function at $x_1$, then $g_1(F_1)=x_1^2-\sigma_x^2 <0$ and $g_2(F_1)=-I_0\left(\sqrt{2}B h_E x_1\right)+E_{req}<0$, hence Slater's condition holds. From the Lagrangian theorem, we conclude that strong duality holds and there exist $\lambda_1\geq 0$ and $\lambda_2 \geq 0$ such that the expression for the capacity in (\ref{eq:C_lagrangian}) holds and is achieved also by $F_0$. Moreover, the complementary slackness conditions $\lambda_1g_1(F_0)=0$ and $\lambda_2g_2(F_0)=0$ must hold. This completes the proof. \vspace{-0.1cm}

\section{Proof of Theorem \ref{theo:C_necessary_sufficient_condition1}}
\label{App:proof_C_necessary_sufficient_condition1}
Define $J(F)\!\definedas\!I(F)\!-\!\lambda_1 g_1(F)\! -\!\lambda_2 g_2(F)$, then (\ref{eq:C_lagrangian}) can be written as $C=\sup_{F\in\mathcal{F}_A} \, J(F)$.  From  \cite[Theorem 3]{capacity_Rayleigh_fading_Abou_Faycal2001}, if $\mathcal{F}_A$ is convex, and $J(F)$ is concave and weakly differentiable, then $J'_{F_0}(F)\!\leq \!0$ is a necessary and sufficient condition for $J(F)$ to achieve its maximum at $F_0$, where $J'_{F_0}(F)\definedas\lim_{\theta\to 0}\left(J((1\!-\!\theta)F_0\!+\!\theta F)\!-\!J(F_0)\right)/\theta$ is the weak derivative of $J(F)$ at $F_0$. In Appendix \ref{App:proof_unique_optimal_distribution}, we established that $\mathcal{F}_A$ is convex and that $J(F)$ is strictly concave in $F$, since $I(F)$ is strictly concave in $F$ and $g_i(F)$ is linear in $F$ for $i\!=\!1,2$. It remains to be proved that $J(F)$ is weakly differentiable and to determine the derivative $J'_{F_0}(F)\!=\! I'_{F_0}(F)\!-\!\lambda_1 g'_{1,F_0}(F) \!-\!\lambda_2  g'_{2,F_0}(F)$. In \cite[Proof of Theorem 3]{capacity_Rayleigh_fading_Abou_Faycal2001}, it is shown that $I'_{F_0}(F)$ exists and is given by $I'_{F_0}(F)\!=\!\int i(x;F_0)\dd F(x)\!-\!I(F_0)$. It is also shown that for any linear constraint function $g_i(F)$, the derivative is $g'_{i,F_0}(F)\!=\!g_i(F)\!-\!g_i(F_0)$. From the complementary slackness conditions, $g_i(F_0)=0$ must hold since otherwise constraints C1 and C2 will be inactive. Hence, the condition $J'_{F_0}(F)\leq 0$ for the optimality of $F_0$ is $\int i(x;F_0)\dd F(x)\!-C-\lambda_1 g_{1}(F)-\lambda_2 g_{2}(F) \!\leq\! 0$, which reduces to (\ref{eq:C_necessary_sufficient_condition1}). This completes the proof.
\vspace{-0.5cm}
\section{Proof of Corollary \ref{corol:C_necessary_sufficient_condition2}}
\label{App:proof_C_necessary_sufficient_condition2}
We start with the necessary and sufficient condition in (\ref{eq:C_necessary_sufficient_condition1}) that guarantees the optimality of $F_0$. From Appendix \ref{App:proof_C_necessary_sufficient_condition1},  (\ref{eq:C_necessary_sufficient_condition1}) can be written as $\int \!i(x;F_0)\dd F(x)\! -\! C-\sum_{i=1}^2\lambda_i\!g_i(F)\!\leq \!0$. For convenience, we define $g_i(F)\!=\!\int A_i(x)\dd F(x)\! -\!a_i$, $i\!=\!1,2$. Hence, $A_1(x)\!=\!x^2$, $a_1\!=\!\sigma_x^2$, $A_2(x)\!=\!-I_0\left(\sqrt{2}Bh_Ex\right)$, and $a_2\!=\!-E_{req}$. Thus,  (\ref{eq:C_necessary_sufficient_condition1}) can be written as \vspace{-0.2cm}
\begin{equation}
\int \Big(i(x;F_0)-\sum_{i=1}^2\lambda_iA_i(x)\Big)\dd F(x) \leq C-\sum_{i=1}^2\lambda_ia_i.
\label{eq:necessary_condition1_general}
\end{equation}\vspace{-0.1cm}
Next, we prove that (\ref{eq:necessary_condition1_general}) holds if and only if
\begin{equation}
\textstyle{\quad \,\,i(x;F_0)\!\leq\! C\!+\!\sum_{i=1}^2\lambda_i\left(A_i(x)\!-\!a_i\right), \quad \forall x\in[-A,A]},\\[-1ex]
\label{eq:necessary_condition2_general1}
\end{equation}
\vspace{-0.3cm}and\vspace{-0.15cm}
\begin{equation}
\textstyle{i(x;F_0)= C+\sum_{i=1}^2\lambda_i\left(A_i(x)-a_i\right), \quad \forall x\in E_0}.\\[-1ex]
\label{eq:necessary_condition2_general2}
\end{equation}
Clearly, if both conditions (\ref{eq:necessary_condition2_general1}) and (\ref{eq:necessary_condition2_general2}) hold, $F_0$ must be optimal because the necessary and sufficient condition in (\ref{eq:necessary_condition1_general}) is satisfied. Thereby, the converse remains to be proved, i.e., if (\ref{eq:necessary_condition1_general}) holds, (\ref{eq:necessary_condition2_general1}) and (\ref{eq:necessary_condition2_general2}) must also hold. We prove this by contradiction. Assume that  (\ref{eq:necessary_condition1_general}) holds but  (\ref{eq:necessary_condition2_general1}) not. It means that $\exists\, \tilde{x}\in[-A,A]$ such that $i(\tilde{x};F_0)> C+\sum_{i=1}^2\lambda_i\left(A_i(\tilde{x})-a_i\right)$. Now, let $F$ be the unit-step function at $\tilde{x}$, then the left hand side of (\ref{eq:necessary_condition1_general}) becomes $i(\tilde{x};F_0)-\sum_{i=1}^2\lambda_iA_i(\tilde{x})>C-\sum_{i=1}^2\lambda_ia_i$, which violates (\ref{eq:necessary_condition1_general}). Hence, if (\ref{eq:necessary_condition1_general}) holds, (\ref{eq:necessary_condition2_general1}) must also hold. Now, assume that (\ref{eq:necessary_condition1_general}) holds but (\ref{eq:necessary_condition2_general2}) not. That is, assume that for a subset of $E_0$ defined as $E'\!\subset \!E_0$, with positive measure, i.e., $\int_{E'}\dd F_0(x)\!=\!\delta\!>\!0$, (\ref{eq:necessary_condition2_general2}) does not hold. Then, from (\ref{eq:necessary_condition2_general1}), $i(x;F_0)\!< \!C\!+\!\sum_{i=1}^2\!\lambda_i\!\left(A_i(x)\!-\!a_i\right)\!,\, \forall\,x\!\in\! E'$. Now, we can write \vspace{-0.3cm}
\begin{equation}
\begin{aligned}
C-\sum\limits_{i=1}^2 \lambda_i a_i&=I(F_0)-\sum\limits_{i=1}^2\lambda_i\int A_i(x)\dd F_0(x)\\[-1ex]
&=\int\Big(i(x;F_0)-\sum\limits_{i=1}^2\lambda_iA_i(x)\Big)\dd F_0(x),
\end{aligned}
\label{eq:expression_to_violate}
\end{equation}
where we used $C\!=\!I(F_0)$ and that constraints C1 and C2 in (\ref{eq:capacity_problem}) are satisfied with equality for the optimal distribution $F_0$. Since $F_0$ has points of increase on $E_0$ only, we have $\int_{E_0}\dd F_0(x)=\int_{E'}\dd F_0(x)\!+\!\int_{E_0-E'}\!\dd F_0(x)\!=\!\delta\!+\!(1\!-\!\delta)\!=\!1$. Hence, (\ref{eq:expression_to_violate}) reads\vspace{-0.2cm}
\begin{equation*}
\begin{aligned}
&C-\sum\limits_{i=1}^2 \lambda_i a_i=\underset{x\in E'}{\int}\underbrace{\Big(i(x;F_0)-\sum\limits_{i=1}^2\lambda_iA_i(x)\Big)}_{<C-\sum_{i=1}^2 \lambda_i a_i}\dd F_0(x)\\[-1ex]
&+\hspace{-0.3cm}\underset{x\in E_0-E'}{\int}\underbrace{\Big(i(x;F_0)-\sum\limits_{i=1}^2\lambda_iA_i(x)\Big)}_{=C-\sum_{i=1}^2 \lambda_i a_i}\dd F_0(x)
<C-\sum\limits_{i=1}^2 \lambda_i a_i
\end{aligned}
\label{eq:expression_to_violate2}
\end{equation*}
which is a contradiction. Hence, if (\ref{eq:necessary_condition1_general}) holds, (\ref{eq:necessary_condition2_general2}) must also hold. Therefore, (\ref{eq:necessary_condition2_general1}) and (\ref{eq:necessary_condition2_general2}) are  necessary and sufficient conditions for the optimality of the input distribution $F_0$. Next, we obtain  condition (\ref{eq:C_necessary_sufficient_condition2}) from (\ref{eq:necessary_condition2_general1}) and (\ref{eq:necessary_condition2_general2}). By definition, the marginal information density $i(x,F_0)$ is given by \cite{SMITH19712}
\begin{equation*}
\begin{aligned}
&i(x,F_0)=\int_y p(y|x)\log_2\left(\frac{p(y|x)}{p(y;F_0)}\right)\dd y\\[-0.8ex]
&=\int p(y|x)\log_2(p(y|x))\dd y-\int p(y|x)\log_2(p(y;F_0))\dd y\\[-0.8ex]
&=-\frac{1}{2}\log_2(2\pi\e\sigma_n^2)-\!\!\frac{1}{\sqrt{2\pi\sigma_n^2}}\int\e^{-\frac{(y-xh_I)^2}{2\sigma_n^2}}\log_2(p(y;F_0))\dd y 
\end{aligned}
\label{eq:i_x_F0}
\end{equation*}
where the first term on the right-hand side is the negative of the entropy of the noise. Finally, using the definitions of $A_i(x)$ and $a_i$ for $i=1,2$, (\ref{eq:necessary_condition2_general1}) and (\ref{eq:necessary_condition2_general2}) reduce to (\ref{eq:C_necessary_sufficient_condition2}). This completes the proof.
\section{Proof of Theorem \ref{theo:discrete_optimal_distribution}}
\label{App:proof_discrete_optimal_distribution}
Our proof of the discreteness of the optimal input distribution parallels that in \cite[Section IV]{Hermite_bases_Abou_Faycal_2012}. In particular, we show that the equality in (\ref{eq:C_necessary_sufficient_condition2}) cannot be satisfied on a set of points that has an accumulation point, which indicates that the set $E_0$ must be discrete and the optimal input $X$ must be a discrete random variable. We start with the necessary and sufficient conditions for the  optimality of $F_0$ in (\ref{eq:C_necessary_sufficient_condition2}) and extend it to the complex domain, then the LHS of (\ref{eq:C_necessary_sufficient_condition2}) reduces to 
\begin{IEEEeqnarray}{ll}
s(z)=\lambda_1\!\left(z^2-\sigma_x^2\right)\!-\!\lambda_2\left(I_0\left(\sqrt{2}Bh_E z\right)\!-\!E_{req}\right)+C \nonumber\\
+\frac{1}{2}\!\log_2(2\pi\e\sigma_n^2)\!+\!\!\frac{1}{\sqrt{2\pi\sigma_n^2}}\!\!\int\!\!\!\e^{-\frac{(y-zh_I)^2}{2\sigma_n^2}}\log_2(p(y;F_0))\dd y, \quad\,\,\,\,
\label{eq:sz}
\end{IEEEeqnarray}
where $z\in\mathbb{C}$. The extension to the complex domain is necessary to use the identity theorem for analytic functions in complex analysis. In particular, the function $s(z)$ is analytic over the complex domain, since the quadratic function, the modified Bessel function, and the exponential function are all analytic \cite{SMITH19712}. A necessary condition for the optimal input distribution to be $F_0$ is that $s(z)$ must be zero $\forall \,z\in E_0$. But from the identity theorem, if the set $E_0$ has an accumulation point and the analytic function $s(z)=0,$ $\forall z\in E_0$, then $s(z)$ is necessarily zero over the whole complex domain, i.e., $s(z)=0$, $\forall z\in \mathbb{C}$. Next, we show that $s(z)$ cannot be zero, $\forall\, z\in \mathbb{C}$, which implies that $E_0$ cannot have an accumulation point, i.e., $E_0$ must be discrete. 

First, similar to \cite{Hermite_bases_Abou_Faycal_2012}, we set $\sigma_n^2\!=\!1$ to simplify the proof without loss of generality and express the last integral term in (\ref{eq:sz}) in terms of the Hermite polynomials $H_m(y)$ defined in \cite[Appendix F]{Hermite_bases_Abou_Faycal_2012}. In particular, since $\log_2(p(y;F_0))$ is a continuous function of $y$ and is square integrable with respect to  $\e^{-y^2/2}$,  it can be written in terms of the Hermite bases as\vspace{-0.2cm} 
\begin{equation}
\log_2(p(y;F_0))=\sum\limits_{m=0}^\infty c_mH_m(y),\vspace{-0.2cm}
\label{eq:logarithm_Hermite}
\end{equation}
where $c_m$ are constants. Hence, the last term of $s(z)$ in (\ref{eq:sz}) can be written as
\begin{IEEEeqnarray}{ll}
Z&=\frac{1}{\sqrt{2\pi}}\int\e^{-\frac{y^2}{2}}\e^{-\frac{(h_Iz)^2}{2}+h_I zy}\log_2(p(y;F_0)) \dd y \nonumber\\[-1ex]
&=\frac{1}{\sqrt{2\pi}}\int\e^{-\frac{y^2}{2}}\sum\limits_{n=0}^\infty \frac{(h_I z)^n}{n!}H_n(y)\sum\limits_{m=0}^\infty c_mH_m(y) \dd y,\quad
\label{eq:I}
\end{IEEEeqnarray}
where we used the Hermite polynomial expansion $\e^{-\frac{(h_I z)^2}{2}+h_Iz y}\!=\!\sum_{n=0}^\infty \frac{(h_I z)^n}{n!}H_n(y)$  \cite{Hermite_bases_Abou_Faycal_2012}. Next, using the orthogonality property of the Hermite polynomials with respect to $\e^{-y^2/2}$ given by $\int_{-\infty}^\infty H_n(y)H_m(y) \e^{-y^2/2}\dd y\!=\!m!\sqrt{2\pi}$ if $m\!=\!n$ and zero otherwise \cite[Appendix F]{Hermite_bases_Abou_Faycal_2012}, then $Z$ in (\ref{eq:I}) reduces to $Z\!=\!\sum_{m=0}^\infty c_m(h_I z)^m$. Furthermore, using the Taylor series expansion of the modified Bessel function given by $I_0(z)=\sum_{m=0}^\infty \frac{(z/2)^{2m}}{(m!)^2}$, the Bessel function in (\ref{eq:sz}) can be written as $I_0\left(\sqrt{2}Bh_E z\right)=\sum_{m=0}^\infty a_m z^{2m}$, where $a_m=\frac{(Bh_E/\sqrt{2})^{2m}}{(m!)^2}$. Hence, from (\ref{eq:sz}), $s(z)=0$ reduces to 
\begin{IEEEeqnarray}{ll}
\sum\limits_{m=0}^\infty c_m h_I^m z^m
&=\lambda_2\Big(\sum\limits_{m=0}^\infty a_m z^{2m} - E_{req}\Big)\nonumber \\
&-\lambda_1(z^2-\sigma_x^2)-C-\frac{1}{2}\log_2(2\pi\e).
\label{eq:sz_0}
\end{IEEEeqnarray}
Equating the coefficients of $z^m$, we get
\begin{equation}
\begin{aligned}
c_0&=\lambda_2 (a_0-E_{req})+\lambda_1\sigma_x^2-C-0.5\log_2(2\pi\e); \, c_{\rm odd}=0,\\[-1ex]
c_2&=(\lambda_2 a_1-\lambda_1)/h_I^2;\quad  \quad \!\!\, c_m=\lambda_2 a_{\frac{m}{2}}/h_I^m,\,\, \forall\, {\rm even}\, m\geq 4.
\end{aligned}
\label{eq:cm}
\end{equation}
Inserting (\ref{eq:cm}) into (\ref{eq:logarithm_Hermite}), the output pdf reduces to \begin{equation}
p(y;F_0)=\e^{\ln(2)\sum_{n=0}^\infty c_{2n}H_{2n}(y)}.
\label{eq:pdf_output}
\end{equation}
Next, we consider two cases based on whether or not the EH constraint is  active. We will show that in both cases, the optimal input distribution is  discrete with finite number of mass points. \\
\emph{Case 1 ($\lambda_2\!=\!0$):} If the EH constraint is inactive, i.e., C2 in (\ref{eq:capacity_problem}) is satisfied with strict inequality, then  $\lambda_2=0$ from the complementary slackness, cf. Theorem \ref{theo:unique_distribution}. In this case, the coefficients in (\ref{eq:cm}) reduce to $c_0=\lambda_1\sigma_x^2-C-0.5\log_2(2\pi\e)$, $c_2=-\lambda_1/h_I^2$, and $c_m=0,\,\forall\,m\neq\{0,2\}$. Using the Hermite polynomials $H_0(y)=1$, $H_2(y)=y^2-1$  \cite[Appendix F]{Hermite_bases_Abou_Faycal_2012}, the output pdf in (\ref{eq:pdf_output}) reduces to 
\begin{equation}
p(y;F_0)=\e^{\ln(2)(c_{0}-c_2)}\e^{\ln(2)c_2 y^2}.
\label{eq:pdf_AP_PP}
\end{equation}
Since the support of $p(y;F_0)$ is the whole real line $\mathbb{R}$ and $c_2<0$, the output distribution in (\ref{eq:pdf_AP_PP}) is the  Gaussian distribution with zero mean. Now, for $Y$ to be Gaussian distributed in the AWGN channel model $Y=Xh_I+N$, then $X$ must also be Gaussian distributed. However, with the PP constraint $|X|\leq A$, $X$ cannot be Gaussian distributed on a bounded interval. Thus, the output distribution in (\ref{eq:pdf_AP_PP}) is invalid. Hence, the assumption $s(z)=0$ used to obtain (\ref{eq:sz_0}) cannot hold over the whole complex plane, and from the identity theorem, $E_0$ cannot have an accumulation point and the optimal distribution of $X$ must be discrete. \\
\emph{Case 2 ($\lambda_2>\!0$):} In this case, the EH constraint is active, i.e., C2 in (\ref{eq:capacity_problem}) is satisfied with equality and the coefficients $c_m$ are given by (\ref{eq:cm}). From \cite[Appendix F]{Hermite_bases_Abou_Faycal_2012}, the Hermite polynomials of even orders are function of even powers of $y$. Thus, the output distribution in (\ref{eq:pdf_output}) reduces to \vspace{-0.2cm}
\begin{equation}
p(y;F_0)=\e^{\ln(2)\sum\limits_{n=0}^\infty q_{n}y^{2n}}=\prod_{n=0}^\infty\e^{\ln(2)q_ny^{2n}},
\label{eq:py}\vspace{-0.2cm}
\end{equation}
where $q_n$ are non-zero constants. It can be easily verified that for some $n\to\infty$, $\exists\, q_n>0$, in which case $p(y;F_0)$ in (\ref{eq:py}) cannot be a valid distribution since it is unbounded. Hence, we conclude that $s(z)=0,\forall\,z\in\mathbb{C}$, cannot hold and $E_0$ cannot have an accumulation point and must be discrete. 
Finally, the finiteness of the number of mass points in $F_0$ follows from the PP constraint $|X|\leq A$. In particular, a bounded set of discrete points must be finite. This completes the proof. 
\bibliographystyle{IEEEtran}
\bibliography{literature}

\end{document}